\documentclass[10pt, conference, letterpaper]{IEEEtran}

\usepackage{graphicx}
\usepackage{url}
\usepackage{array}
\usepackage{multirow}
\usepackage{caption}
\DeclareCaptionType{copyrightbox}
\usepackage{amssymb}
\usepackage{amsmath}
\usepackage{subfigure}
\usepackage{bm} 
\usepackage{times}
\usepackage{algorithm}
\usepackage{algpseudocode}
\usepackage{cite}
\usepackage{color}
\usepackage[normalem]{ulem}
\usepackage{comment}
\usepackage{balance}
\usepackage{amsthm}

\usepackage{xspace}

\newcommand{\ve}[1]{\mathbf{#1}}

\newcommand{\algnameDT}{DISTRO\xspace}
\newcommand{\algnameD}{{\sc DISTRO}\xspace}

\def\squareforqed{\hbox{\rlap{$\sqcap$}$\sqcup$}}
\def\qed{\ifmmode\squareforqed\else{\unskip\nobreak\hfil
\penalty50\hskip1em\null\nobreak\hfil\squareforqed
\parfillskip=0pt\finalhyphendemerits=0\endgraf}\fi}

\newcolumntype{C}[1]{>{\centering\let\newline\\\arraybackslash\hspace{0pt}}m{#1}}
\newcolumntype{L}[1]{>{\raggedright\let\newline\\\arraybackslash\hspace{0pt}}m{#1}}

\usepackage{pifont}

\usepackage{booktabs}

\newtheorem{innercustomthm}{Theorem}
\newenvironment{customthm}[1]
  {\renewcommand\theinnercustomthm{#1}\innercustomthm}
  {\endinnercustomthm}

\newtheorem{theorem}{Theorem}

\newtheorem{myprop}{Proposition}

\newcommand{\qedwhite}{\hfill \ensuremath{\Box}}

\begin{document}

\newfont{\titlefont}{phvb8t at 17pt}
\twocolumn[
\centerline{\titlefont Scalable Real-time Transport of Baseband Traffic}
\vspace{20pt}
\centerline{Krishna C. Garikipati \quad Kang G. Shin}
\vspace{10pt}
\centerline{University of Michigan}
\vspace{12pt}
]
\medskip

\thispagestyle{plain}
\pagestyle{plain}



\begin{abstract}

In wireless deployments, such as Massive-MIMO, where radio front-ends
and back-end processing are connected through a transport network, meeting
the real-time processing requirements is essential to realize the
capacity gains from network scaling. While simple forms of baseband
transport have been implemented, their real-time analysis at much larger
scale is lacking.

Towards this, we present the design, delay, and capacity analysis
of baseband transport networks, utilizing results from real-time
systems in the context of wireless processing. We propose a novel
Fat-Tree-based design, called \algnameD, for baseband transport, which is
a real-time network that bounds the maximum end-to-end
transport delay of each baseband packet. It achieves this by placing
design constraints and bounding the queuing delay at each aggregation
point in the network. We further characterize the wireless capacity
using \algnameD and provide an efficient search algorithm for the design
of a capacity-achieving baseband transport.




%
\end{abstract}
\section{Introduction}\label{intro}

In emerging wireless architectures such as C-RAN \cite{cran, cran_survey},
and Massive-MIMO \cite{massive, argos}, a baseband transport network connects
radios (front-ends) to a backend processing infrastructure. For instance, in a
C-RAN deployment, a fronthaul network carries baseband samples between the
remote radios and the baseband processors located in a datacenter. Similarly,
in massive MIMO, tens of antennas are connected to a common baseband processor
through a transport network. Since the real-time capability of baseband transport
network is the key enabler of such architectures, its design determines both the
scale and achievable wireless capacity of the deployment.

An instance of baseband transport network was realized in ARGOS \cite{argos},
where 64 antennas are connected over Ethernet for centralized processing.
While ARGOS' transport design is intuitive, its real-time analysis or applicability
in a larger network remains unknown. Clearly, a baseband transport network
should support traffic aggregation from the radios, for MIMO processing \cite{massive}, or for
resource pooling in C-RAN \cite{oai:15}, where a common compute platform decodes data/signals
for multiple base-stations. A tree-based design using aggregation switches fits the criteria,
which, coupled with packet scheduling policies, form the basis of real-time transport.

Beyond the primitives of aggregation and real-time delivery, the
design of a baseband transport network is ultimately guided by the following
key requirements:

\textbf{Guarantees}. The real-time nature of the wireless sample processing imposes
stringent constraints on the transport network requiring end-to-end (e2e) guarantees
for delay and jitter. For instance, the transport delay bound can be as low as
few microseconds for WiFi samples \cite{sora}, to few hundred microseconds for
LTE samples \cite{oai:15}. Moreover, the transport behavior of the radio samples
must be predictable, i.e., given a network topology and the traffic sources, one
must be able to model the delivery of the baseband traffic. This model is necessary
for an e2e schedulability analysis --- determining whether the given network can
meet the requested delay bounds. However, modeling packet traffic in the network core is known 
(e.g., from the real-time systems literature) to be intractable
due to the non-periodic nature of arrivals \cite{hui:95, survey:94}. Therefore,
in a general baseband transport network consisting of multiple switches, it is
not entirely clear which packet scheduling policies (implemented by switches)
will achieve e2e schedulability.

\textbf{Scalablility}. Considering the size of future radio deployments, the baseband transport
design should be scalable. It must be extensible to any number of radios while
preserving its predictable behavior. Also, it must require only few additional resources
(e.g., cables, switches) to add a large number of radios to the network. The scalability
of baseband transport is desirable, if not necessary, for massive MIMO systems which are
equipped with tens or hundreds of antennas/radios. Existing datacenter
designs \cite{alfaras, fat_tree_follow} that are optimized for scalability, are likely candidates
for baseband transport. However, they are primarily designed to handle bisection
traffic between the compute clusters, whereas traffic flows in a baseband network
are exclusively in the north--south direction.


\textbf{Optimality}. Whether the baseband traffic is schedulable or not depends
on the transport rate, which, in turn, depends on the sample quantization widths
used at the radios. As the selected quantization widths affect the wireless capacity
through quantization noise \cite{spiro, uta}, there is an indirect dependency between
schedulability of the traffic and the wireless capacity
of the network. Ideally, the network should operate at a point that ensures schedulability
but also maximizes wireless capacity.

To meet the above requirements, we propose \algnameDT, a design for real-time
baseband transport (such as fronthaul) networks that can potentially scale to a
large number of radios. \algnameD utilizes a logical tree structure of radio front-ends
and network switches: the radios represent the leaf nodes of the tree, the network
switches represent the internal nodes, and the root of the tree is a common aggregation
point that connects to a pool of baseband processors. \algnameD's design
supports real-time transport as it allows us to bound the maximum transport delay
of each baseband packet. Specifically, using a constrained tree design, the
upper bound on the waiting time at any switch can be obtained irrespective of
the input arrival sequence, from which one can obtain the maximum total delay of
a baseband flow. As a result, the network switches can utilize schedulability results
from the real-time systems literature and implement scheduling policies
(such as EDF \cite{zheng:94} and fixed-priority \cite{cmu:90}) to achieve e2e
delay guarantees.

Since scheduling policies are subject to the baseband traffic parameters, any
addition of radios or changes in them requires policy changes in the entire
network, making the design unscalable. Therefore, the tree structure in \algnameD
is partitioned into aggregation- and edge-switch networks; the schedulability
analysis is done only at the edge-switch network whereas the default packet
scheduling (FIFO) is implemented in the aggregation-switch network. This
logical division along with the tree-based design enables transport scaling to
a large number of radios.

\algnameD selects quantization widths to maximize the wireless capacity
while ensuring e2e schedulability. A brute-force search of optimal quantization
widths, however, has exponential complexity in the number of radios. By using the
monotonic dependence of  the wireless capacity and the schedulability on the quantization
widths, we propose a greedy-based approach that is optimal but has much lower
runtime complexity.

In summary, this paper makes the following contributions:

\begin{itemize}

\item \textbf{Design}. Proposal of a Fat-Tree architecture for scalable deployment of
baseband radios;

\item \textbf{Delay}. Calculation of the maximum delay bound of a baseband packet
in the network and its use to achieve e2e schedulability; and

\item \textbf{Capacity}. Characterization of the wireless capacity under the
schedulability constraint, and development of an efficient search algorithm to
maximize the capacity.

\end{itemize}
The rest of this paper proceeds as follows. Sec.~\ref{background}
provides the background on the baseband transport while Sec.~\ref{architecture}
presents our proposed transport design and its evaluation results in a simulated
network scenario. In Sec.~\ref{application}, we obtain the maximum wireless
capacity with end-to-end schedulability through our proposed algorithm. Finally,
Sec.~\ref{related} presents the related work and Sec.~\ref{conclusion} concludes
the paper.
\section{Wireless Baseband Transport}\label{background}

In this section we provide the background on baseband transport and processing,
and describe how real-time scheduling fits into its design.

\subsection{Network Primitives}

\begin{figure}
     \begin{minipage}[t]{0.48\columnwidth}
         \centering
         \includegraphics[ width=\columnwidth]{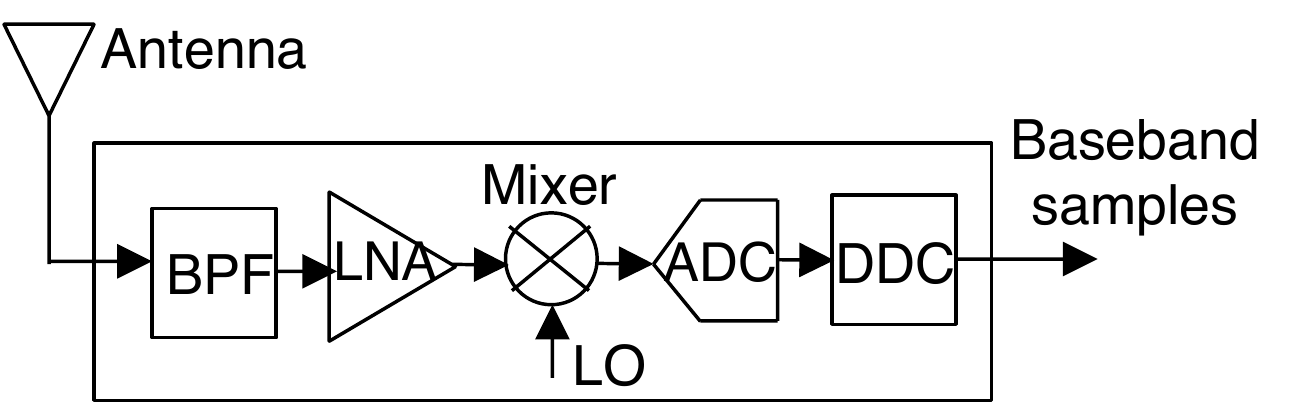}
         \caption{Radio front-end design for baseband conversion.
         \label{fig:distro:radio}}
     \end{minipage}
     \hspace{0.2cm}
     \begin{minipage}[t]{0.48\columnwidth}
         \centering
         \includegraphics[width=\columnwidth]{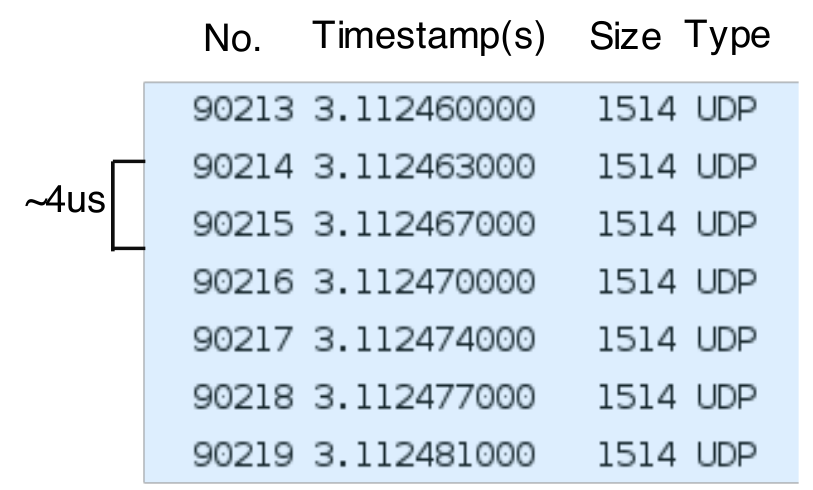}
         \caption{USRP2 transport log at 25MHz sampling rate. 
    \label{fig:distro:usrp}}
     \end{minipage}
 \end{figure}


The radio front-ends in a baseband network act as
converters between RF samples and complex (I and Q) baseband samples.
Fig.~\ref{fig:distro:radio} shows the components of such a radio. In the receive mode,
the RF signal is down-converted, filtered and passed through an analog-to-digital
converter (ADC) that gives out a digitized (e.g., 16-bit) stream of baseband samples.
This stream is broken into fixed-size blocks, which are then transported as payloads
in special-purpose packets generated with appropriate headers and tags.

Suppose there are $n$ radios in the network, where each radio is denoted by index $i \in \{1,\ldots, n\}$.
Let $f_i$ denote the desired sampling frequency from the ADC (achieved through decimation by
digital-down converters), and let $Q_i$ be the number of bits used to represent each
I (and Q) baseband sample. Then, the transport data rate (in bits/s) of radio $i$ can be expressed as:
\begin{align}
R_i = 2Q_if_i. \label{eq:distro:rate}
\end{align}
Further, let $B$ denote the fixed payload size (in bits) of a transport packet.
The inter-packet arrival time (in seconds) at radio $i$, assuming negligible packet
overhead, is given by:
\begin{align}
T_i = \frac{B}{R_i}. \label{eq:distro:arrival}
\end{align}
Since the ADC operates at a fixed frequency, and fixed-size blocks are used, the
packet inter-arrival time is a constant at each radio, which we refer to as the
\emph{period} of the arrival process.

\textbf{Example}:
USRP \cite{usrp} is a common software-radio platform that uses the UDP protocol for baseband transport.
Fig.~\ref{fig:distro:usrp} shows the timestamps from the transport log of a USRP2 running
at 25MHz sampling rate and payload size of 1492 Bytes. Using Eqs.~\eqref{eq:distro:rate}
and \eqref{eq:distro:arrival}, the packet inter-arrival time with 8-bit quantization
is calculated to be $2.98\mu$s. This is indeed close to inter-packet arrival time $\in[3,~4] \mu$s
observed from the USRP2 logs (Fig.~\ref{fig:distro:usrp}). Note that
the logged arrival period has a minimum 1$\mu$s resolution.

Wireless protocols have a fixed e2e processing deadline for PHY-layer
primitives such as channel sensing and decoding. Assuming a fixed (or worst-case)
processing time at the baseband processors, the e2e PHY deadlines impose a
maximum transport delay. Therefore, in order to support real-time processing, the
generated baseband packets from the radios must be transported to the baseband processors
within a fixed amount of time. That is, each radio $i$ has an e2e transport
delay bound, $D_i$, that the transport network must satisfy.

The traffic specification of radio $i$ is thus given by a 2-tuple $\tau_i = (T_i, D_i)$,
which represents the inter-arrival time, and the e2e delay bound, respectively.
The radio traffic is said to be {\em schedulable} if for all $1\leq i \leq n$,
the maximum delay experienced by packet of radio $i$ is not greater than the requested
delay bound $D_i$.

The transport network in large deployments of Cloud-RAN runs
over a fiber infrastructure such as dark fiber\cite{cran_survey}.
While in indoor environments, the radios may be deployed using high-capacity
Ethernet or Infiniband links \cite{argos}. In both scenarios, baseband samples are
exchanged through packet transmissions. Thus, the baseband transport network
can be modeled as a packet-switch network with one or more network switches.
Every packet in the network passes through multiple links, switches, and routers
before reaching its destination. While many routes can exist for a packet,
we assume fixed routing in the network, which is necessary for
e2e delay guarantees \cite{kandlur:91}.

Despite its advantages, packet-switching introduces various delays at each link
in a selected route. The e2e packet delay is the summation of delays over links
and switches along the selected route, which is composed of:
\begin{itemize}
\item Propagation delay ($t_p$): the time taken for the packet to reach the next switch;
\item Switching delay ($t_s$): the time taken for the packet to move from the ingress to egress port of a switch;
\item Transmission delay: the time needed to transmit the packet, which is a function of the link capacity; and
\item Queuing delay($t^q$): the waiting time of packet in a switch's egress queue.
\end{itemize}
\begin{figure}
    \centering
    \includegraphics[width=0.8\columnwidth]{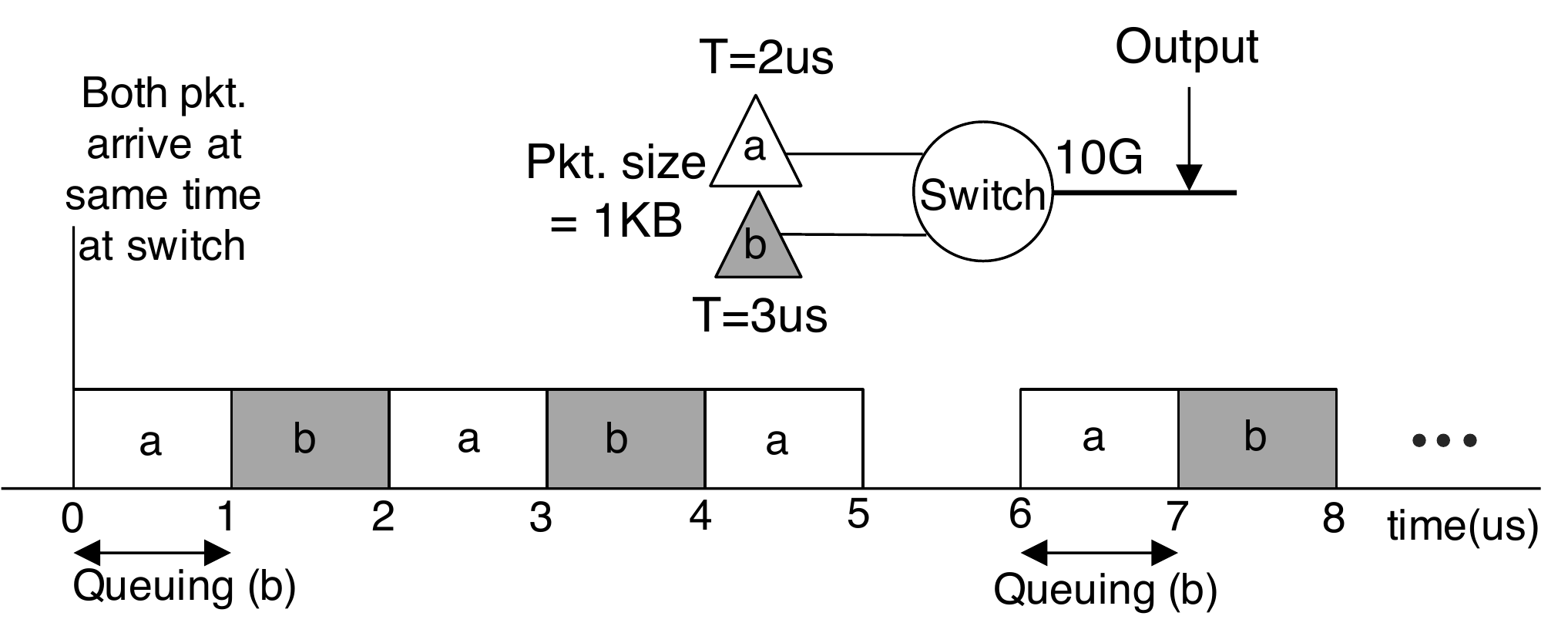}
    \caption{Switch FIFO output for two periodic flows. The output sequence is non-periodic.\label{fig:single_switch}}
\end{figure}
Among the different delays, the queuing delay is the only unknown that
can be different for each packet in the network. In general, it is a
function of the switch's scheduling policy and the input arrival sequence. One
needs to model these delays at each switching stage, which becomes intractable
for large networks as the output sequence from a switch is non-periodic even though packets
arrive periodically. For example, consider a simple baseband network (Fig.~\ref{fig:single_switch})
with periodic baseband traffic from 2 radios having inter-arrival times of $2\mu$s and $3\mu$s,
respectively. Assume the processing link capacity is 10Gbps and the packet size is 1000 bytes.
That is, the output transmission time of both flows is $1\mu$s. Assume packets from
the two radios arrive in the queue at the same time instant in the beginning.
Fig.~\ref{fig:single_switch} shows the timeline of the queue output. As one can see,
the inter-arrival time of the 3.3Gbps flow at the output is non-periodic (inter-arrival
times of 2$\mu$s and 4$\mu$s) that is induced from the waiting in the queue.
This non-linearity of queuing makes it difficult to model e2e packet arrivals,
a fact well-known in the literature \cite{hui:95}.

Common packet-switching techniques such as First-In-First-Out (FIFO) and
Round-Robin (RR), which are designed for best-effort traffic flows, are not suitable
for the real-time traffic that requires e2e delay guarantees. To support
real-time baseband traffic, the switches can use various scheduling policies
that were developed by real-time systems researchers \cite{survey:94, hui:95}.
Among them, the Earliest-Deadline-First (EDF) scheduling is a natural approach where each arriving
packet is assigned a deadline according to the requested delay bound, and the packet
with the earliest deadline is transmitted first. This scheduler is optimal \cite{liu_layland}
in the sense that if packets meet their deadlines using any scheduling policy, so
will they using EDF. While the original EDF considered
implicit deadlines (deadline is the same as the period) and preemption, one
can generalize it further to obtain both necessary and sufficient conditions for
schedulability with arbitrary deadlines. Theorem \ref{theorem:edf} (Appendix)
formally states these conditions for both cases -- with and without preemption.

Deadline scheduling is based on dynamic prioritization and thus difficult to
realize in practice. A more feasible approach is the fixed-priority scheduling
where each traffic flow is assigned a static priority \cite{cmu:90} where incoming packets
are transmitted in the order of their priority. Theorem \ref{theorem:fixed}
(Appendix) shows that there exists a schedulability test to determine whether the set
of traffic flows with given priorities meet their delay bounds. Furthermore, one can do
an iterative or offline search and use the schedulability criteria to arrive at
the feasible priority-assignment policy if one can be found \cite{tindell:94}. The
necessary conditions for schedulability, however, are known only under special
circumstances (e.g., when the deadline is less than or equal to the period).

\subsection{Transport Delay Bound}\label{delay_bound}

The wireless processing design depends on the wireless protocol and the target
architecture. For WiFi signals, where slots are $9\mu$s long, baseband samples
are streamed and decoded on the fly \cite{sora}. In contrast, LTE has 1$ms$-long
subframes that are decoded on an accumulated 1$ms$ buffer of baseband samples
\cite{oai:15}. The time required to decode the baseband samples (or frames) depends
on the capability and the optimizations of the platform. In this paper, we assume
the processing time is fixed, and focus on the transport delays.

The transport delay bound of radio $i$ is computed by subtracting the maximum
processing time, $T_{proc}$, from its end-to-end protocol deadline, $T_{prot}$,
as:
\begin{equation}
D_i = T_{prot} - T_{proc}.
\end{equation}
In case of WiFi signals, the protocol deadline, $T_{prot}$, can be 4$\mu$
(since CCA assert should occur within $4\mu$s during energy sensing \cite{80211ac}).
Assuming $T_{proc} = 2\mu$s to perform sample summation, this results in a $2\mu$s
delay bound for the transport network. On the other hand, for LTE signals, $T_{prot} = 2$ms,
as there is no channel sensing, and the protocol deadline is governed by the HARQ process.
Consequently, the delay bound is much larger (0.5--0.7 ms) than the WiFi case.

It is worth noting that the transport delay bound, $D_i$, is not always fixed
but can vary with the number of radios, $n$, since the processing time, $T_{proc}$,
typically increases with $n$.
For instance, $T_{proc} \propto O(n^3)$, when decoding $n$ spatially multiplexed signals
\cite{bigstation:13}.
\section{DISTRO DESIGN AND DELAY ANALYSIS}\label{architecture}

We now describe the construction, requirements, and delay analysis of
the proposed real-time transport design for baseband traffic.


\begin{figure}
    \centering
    \includegraphics[width=0.9\columnwidth]{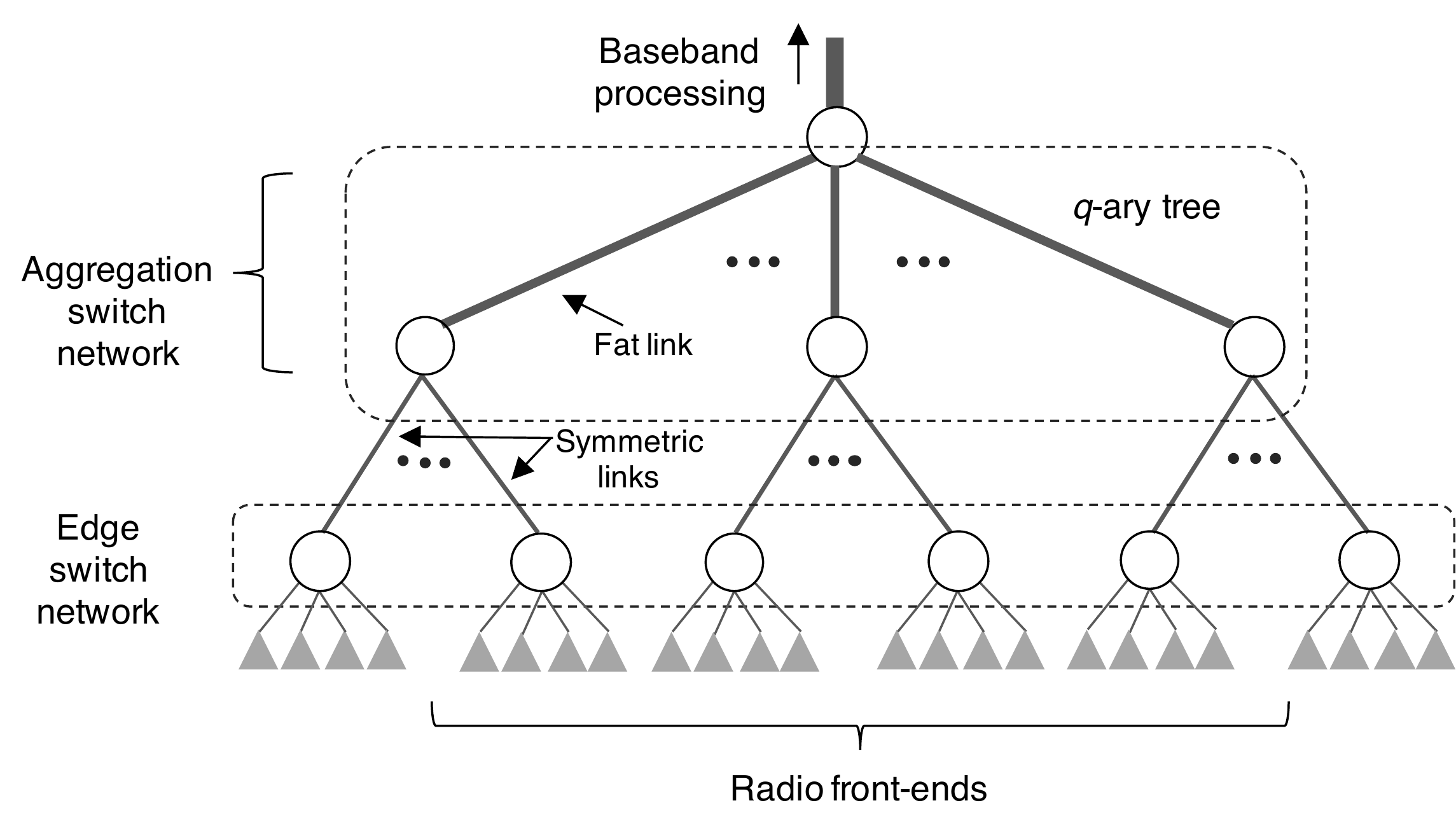}
    \caption{Fat-Tree architecture of \algnameD connecting radio front-ends
    with fat and symmetric links.\label{fig:tree_main}}
\end{figure}

\subsection{Design Philosophy}

The design of a baseband network is driven by two key observations. First, the
baseband traffic flows exclusively between the radios and the processor pool
(cross-traffic between radios is negligible). Second, depending on the application,
baseband traffic is aggregated into one or more links for processing.
In addition, the baseband transport operates in real time: given
the traffic flows and delay bounds, we must be able to give a sufficient, if not
necessary, condition to check their e2e schedulability.


We propose \algnameD, a baseband network design that combines the tree structure
with real-time scheduling. Fig.~\ref{fig:tree_main} illustrates the proposed design
in a deployment of heterogeneous radios. The radios are connected to edge switches,
and the links from the edge switches are aggregated at multiple levels till the
root switch. The destination is assumed to be located at the root switch,
which is the common aggregation point for the baseband packets. The traffic
destination is assumed to be a physical point that is connected to a pool of
baseband processors.

\algnameD's design accommodates an increasing number of radios
without severely affecting their delay performance, and makes the
schedulability flexible to the addition and removal of radios. This is achieved
by partitioning the baseband network into two components: edge- and aggregation-switch
network. The edge switch network contains the edge switches that form the first entry
point of a baseband flow. From a deployment standpoint, each edge switch could connect
a group of radios that are in physical proximity of each other, for instance, a
basestation site in a cellular network. The aggregation network contains all the remaining
switches except the edge switches, and its purpose is to aggregate traffic from edge
towards the destination.


\subsection{Design Requirements}

The aggregation network in \algnameD is a logical tree of links and switches.
To simplify the schedulability analysis, we place the following restrictions on its design.


\noindent \textbf{1) Fat-Tree}. A tree is a Fat-Tree if for every switch in the
tree, the switch's uplink capacity is greater than or equal to the sum capacity
of the incoming links.



\noindent \textbf{2) Symmetric}. A tree is symmetric if for every switch, interchanging
the incoming links yields the same tree.

\noindent \textbf{3) Non-preemptive}. The network switches always use a non-preemptive
scheduling policy, i.e., an ongoing packet transmission is never preempted.

\noindent \textbf{4) Equal packet sizes}. The basesband packet sizes in the
network are always equal irrespective of their source.


\subsection{End-to-End Guarantees}\label{sec:guarantees}

As noted earlier, a packet can experience queuing delay at any of the switches 
in the network. Beyond the edge switch, characterizing the queuing behavior becomes 
difficult as the packet arrivals are no longer periodic \cite{hui:95}. However, 
under our symmetric Fat-Tree construction, we can readily bound the maximum 
queuing time.

Assume $K$ edge switches, and let $S_i, i = 1,2, \ldots K$, denote the set of
radios connected to edge switch $k$ where $S_k \subseteq \{1,2 \ldots, n\}$.
For simplicity, assume at time $t=0$ packets from radios of each edge switch
are queued at the switch, and the packets arrive periodically at the switch thereafter.
Further, assume no transmission or propagation delay from the radio to
the edge switch. Let $C_1$ denote the transmission time of a packet on the link
connecting the edge switch and the next aggregation switch. Since packet size, $B$,
is fixed for the network, every baseband packet has the same transmission time
going through the edge switch.

\begin{figure}
    \centering
    \includegraphics[width=0.7\columnwidth]{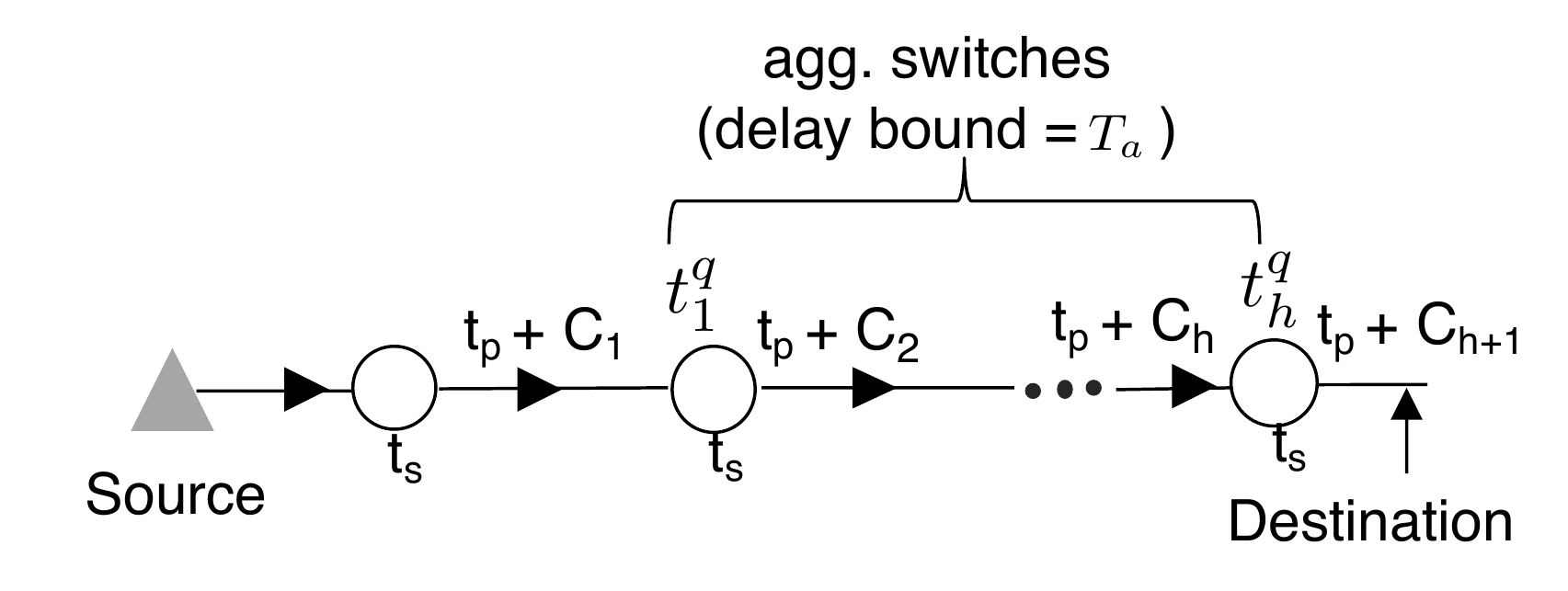}
    \caption{Delay path of baseband packet from the source to the destination 
    including queuing delays.\label{fig:path}}
\end{figure}


%
\begin{theorem}\label{theorem:distro}
In a $q$-ary symmetric fat-aggregation-tree of height $h$, the radio traffic $\tau_i = (T_i, D_i), i = 1,2,\ldots, n$,
is schedulable, if for every $k, 1 \leq k \leq K$, the set of traffic flows
$(T_i, d_i^{'}), \forall i \in S_k$, with transmission time $C_1$ and no preemption,
is schedulable at edge switch $k$, where $d_i' = D_i - \frac{1-q^{-h}}{1-q^{-1}}C_1 - (h+1)(t_s + t_p)$,
$\forall i$.
\end{theorem}
\begin{proof}\let\qed\relax


We show the proof by construction. Let the transmission time sequence as a 
packet traverses from the edge to the destination be represented 
as $C_1, C_2, \ldots, C_{h+1}$, where 
$h$ is the height of the aggregation tree. We denote $h=1$ if the network 
contains only one aggregation switch.

From the symmetric Fat-Tree definition, it holds that if $C_j$ is the transmission
time on the incoming link, then the transmission time on the aggregation link, 
$C_{j+1}$, for all $j, 1 \leq j \leq h$, satisfies:
\begin{equation}
 C_{j+1} \leq C_{j}/q. \label{eq:fatness}
\end{equation}
Assume the edge switches use some non-preemptive scheduling policy whereas 
the aggregation switches use FIFO scheduling. Under non-preemptive scheduling,
the inter-arrival time of packets on the outgoing link of the edge switch is 
equal to or larger than $C_1$. Consider the first aggregation switch: the 
outgoing transmission time is $C_2 (C_2 \leq C_1/q)$. The maximum queuing time 
of a packet at the switch is $(q-1)C_2$ which occurs when $q$ packets, one packet 
from each link, arrive at the same time. Note that since the inter-arrival time 
$C_1 \geq qC_2$, and FIFO scheduling is used, two packets from the same link cannot 
be transmitted unless queued packets from all the other links are transmitted.

This is illustrated in Fig.~\ref{fig:equal} in the case of two links, where 
packet $c$ arrives at the same instant as packet $b$, and is blocked by the 
transmission time of packet $b$.

Continuing along the path in Fig.~\ref{fig:path}, for every $j$, packets in 
incoming links with transmission  time $C_j$ and inter-arrival time greater 
than $C_j$, will have a maximum queuing delay: 
\begin{equation}
 t^{q}_{j} \leq (q-1)C_{j+1}. \label{eq:max_queing}
\end{equation}

Thus, the maximum total delay of \emph{any} baseband packet across the
aggregation network starting from the edge switch is bounded by:
\begin{eqnarray}
T_a &=& \sum_{j=1}^{h} ( t_s + C_{j+1} + t^{q}_{j} + t_p)\\
&\leq& \sum_{j=1}^{h} ( t_s + C_{j+1} + (q-1)C_{j+1} + t_p)\\
&=&h(t_s + t_p) + \sum_{j=1}^{h} qC_{j+1}\\
&\leq&h(t_s + t_p) + \sum_{j=1}^{h} q\frac{C_1}{q^{j}} \\
&=& h(t_s + t_p) + \frac{(1-q^{-h})C_1}{1-q^{-1}} \label{eq:max_delay}
\end{eqnarray}
where we use the inequality $C_{j+1} \leq C_{j}/q\leq \ldots \leq C_{1}/q^{j}$ 
using Eq.~\eqref{eq:fatness}.

Therefore, the aggregation delay bound obtained is independent of the arrivals at
in the network. Now, in order for a baseband packet of radio $i$ to meet its e2e delay bound
$D_i$, the total packet delay at its edge switch must be less than $D_i$ less the
aggregation delay bound. Adding the switching and propagation delays of the
edge switch ($t_p + t_s$), we arrive at the e2e schedulability of all baseband flows 
in the network by checking the schedulability of baseband traffic at each edge switch.\qedwhite
\end{proof}
\noindent Theorem \ref{theorem:distro} provides only a sufficient condition for schedulability,
but does not affirmatively tell us if a given set of traffic flows are schedulable.
Nevertheless, the result is still useful as it allows us to construct a transport network
that guarantees to meet the requested e2e delay bounds. Furthermore, it
gives the scheduling policies that the network must implement: non-preemptive
scheduling at edge switches, and regular FIFO scheduling at aggregation switches.

The edge switches can implement a non-preemptive scheduler which
can either be EDF or fixed-priority. The fixed-priority scheduler is easy to realize
with multiple outbound queues at a switch. Each incoming packet is first classified
and placed in its corresponding queue, and the queues are dequeued in the order
of their priority.

\begin{figure}[t]
    \centering
    \subfigure[Equal tx times]{\label{fig:equal}\includegraphics[width=0.48\columnwidth]{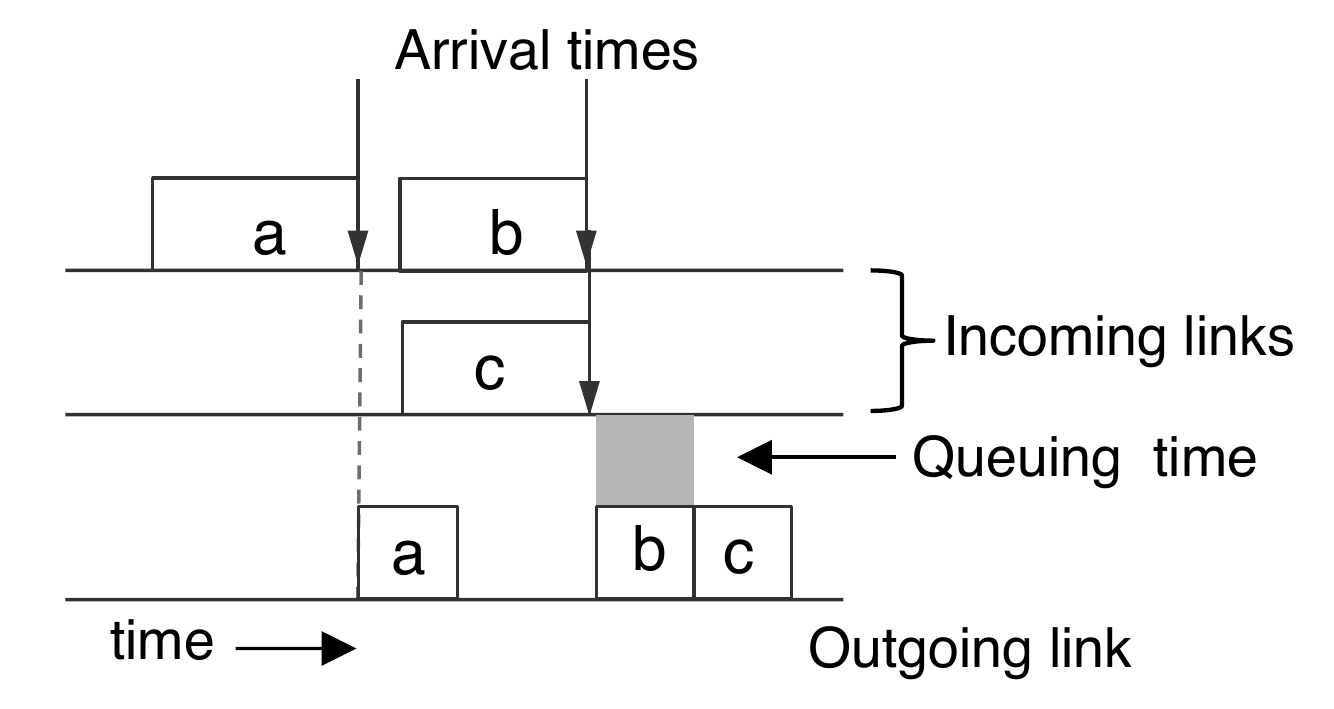}}
    \subfigure[Unequal tx times]{\label{fig:unequal}\includegraphics[width=0.48\columnwidth]{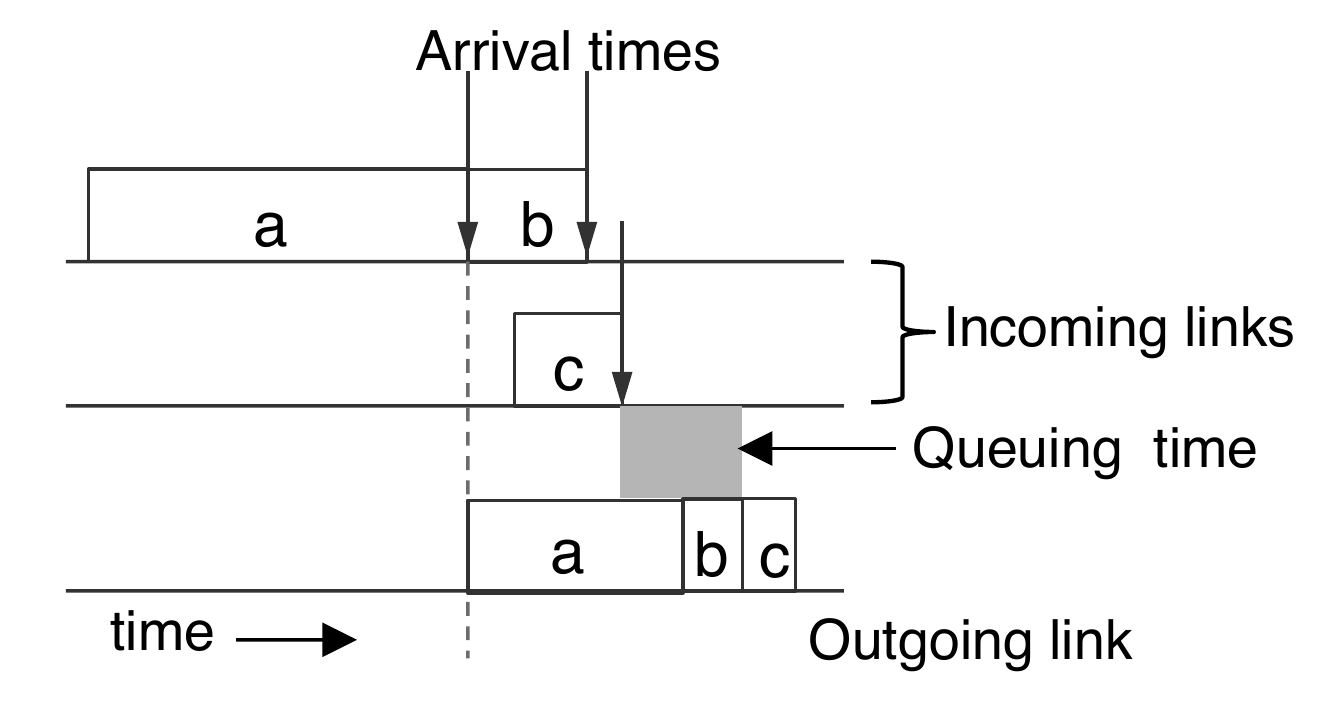}}
    \caption{Blocking of $c$ by at most one packet (left) and more than one packet (right).\label{fig:blocking}}
\end{figure}

The reason for restricting our construction to a symmetric Fat-Tree and non-preemptive
scheduling is simple: the packet transmission (tx) times are equal for the incoming links,
and therefore, the maximum queuing delay in the aggregation network is easy
to obtain. This is not the case, however, if we assume unequal transmission times,
which happens if unequal links are used or when preemption is allowed. For instance,
Fig.~\ref{fig:unequal} shows the blocking of a smaller packet, $c$, due to the
transmission of earlier packets. This example is akin to the head-of-line (HOL)
blocking problem that occurs when subsequent transmissions are held up by the first
transmission \cite{hol}. Then, to arrive at the maximum queuing time in the network,
one must examine different arrival sequences and their transmission times at each
switch, which can quickly become intractable.

\subsection{Scalability}

The design of \algnameD is scalable on three fronts. First, as the number of
radios in the network, $n$, grows, and the number of edge switches, $K$, increases,
the maximum total delay of a packet increases only logarithmically. To see this,
one can upper bound the total aggregation delay in Eq.~\eqref{eq:max_delay}
as follows:
\begin{equation}
T_a < h(t_s + t_p) + 2C_1.
\end{equation}
The delay bound grows linearly with the height, $h$, which is always less than
or equal to $\lceil \log_2 K\rceil$. For large packets (e.g., Jumbo Ethernet frames),
the transmission time is much larger than the switching and propagation time.
In this case, the delay scaling factor, $t_s+t_p$, becomes even less significant.

Second, \algnameD supports heterogeneous baseband radios with differing periods
and deadlines. There is no restriction on the type of radio protocol. This aids
the deployment of a radio access network over multiple wireless standards,
such as LTE and WiFi, at the same physical location.

Finally, each edge switch, $k$, implements a schedulability test locally that is
dependent only on the traffic generated from the  set of radios, $S_k$, connected
to it. Any addition or removal of radios requires one to check only local schedulability,
without disturbing the performance of other flows. This feature enables
incremental deployment of the baseband network.

\subsection{Run-time Scheduling}

The aggregation network in \algnameD implements the default FIFO scheduling.
On the other hand, each edge switch implements a non-preemptive scheduler which
can either be EDF or fixed-priority. The fixed-priority scheduler is easy to realize
with multiple outbound queues at a switch. Each incoming packet is first classified
and placed in its corresponding queue, and the queues are then dequeued in the order
of their priority.

\emph{Background Traffic}. Background traffic such as control messages can disrupt
the real-time performance of the network. To ensure minimum disruption, we
segregate the baseband traffic from the rest of the traffic through flow prioritization.
Background and rest of the traffic are placed in a separate queue,
which has lower priority than baseband traffic. Since background traffic
adds to the non-preemptive delay at each switch, we update the delay bound in
Eq.~\eqref{eq:max_delay} accordingly.

\subsection{Evaluation}\label{evaluation}

We implement and evaluate \algnameD's Fat-Tree architecture using NS-3\cite{ns3}.
NS--3 is a discrete-event network simulator that accurately simulates network
traffic in large deployments. Table~\ref{table:simu_params} shows the simulation
parameters used in our setup. Each edge switch connects 4 heterogeneous radios that
have different flow rates (1, 1.5, 2 and 2.5G) but fixed packet size of 1000 bytes.
We simulate a fat-aggregation-tree with three levels: 1) a core switch that is connected
to the destination through 200Gbps link; 2) $q$ aggregation switches connected to
core switch with 40Gbps links; 3) $q$ edge switches connected to each aggregation
switch with 10Gbps links. Thus, the total number of radios in our setup is $4q^2$.

\begin{table}
\centering
\small{
    \begin{tabular}{|c|c||c|c|}
    \hline
    $t_s$ & 50ns &  Edge--src.   & 10Gbps   \\
    \hline
    $t_p$  & 10ns  & Agg.--Edge & 10Gbps   \\
    \hline
    Pkt. size ($B$) & 1KB & Agg.--Core & 40Gbps \\
    \hline
    $q$-ary tree & 2,3,4  & Core--Dest.  & 200Gbps\\
    \hline
    height ($h$) &  2  & Flow rate ($R_i$) & 1, 1.5, 2, 2.5 Gbps\\
    \hline
    Radios/edge & 4  & Simulation Time & 1sec\\
    \hline
    \end{tabular}
    }
\caption{Simulation parameters}
\label{table:simu_params}
\end{table}
%
%
%
\begin{figure}[t]
    \centering
    \subfigure[Maximum Delay \label{fig:radio_delays_3}]{\includegraphics[width=0.9\columnwidth]{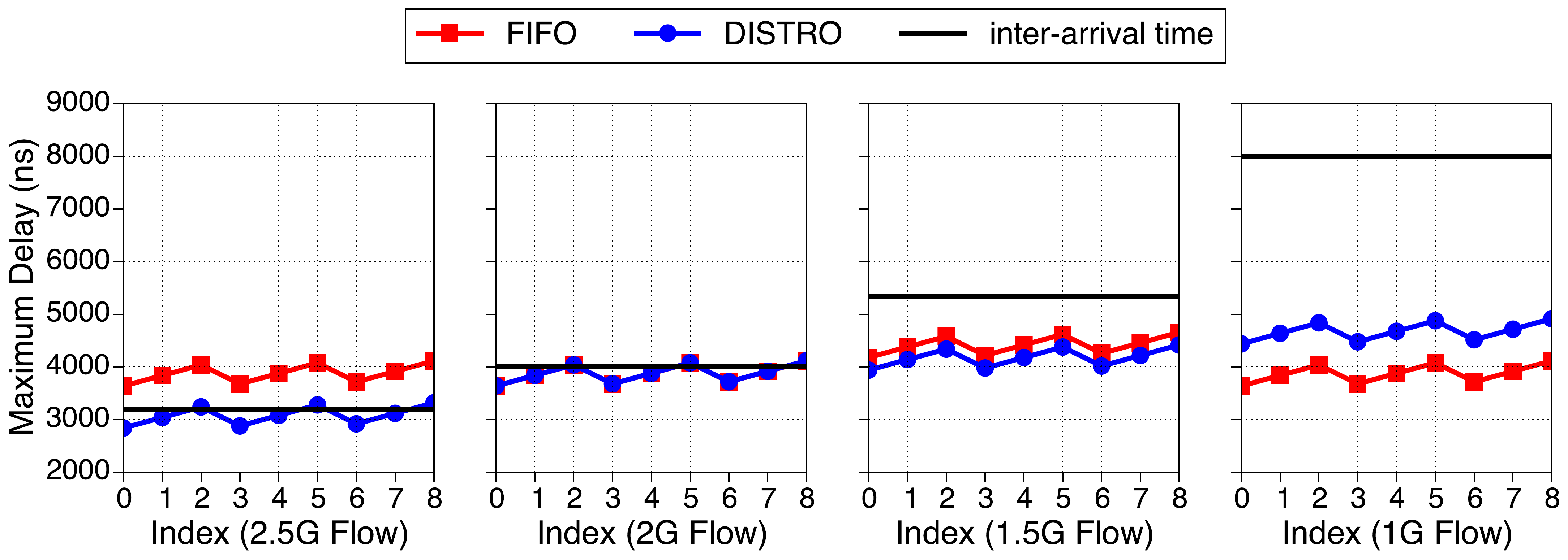}}
    \vspace{1mm}
    \subfigure[Jitter]{\includegraphics[width=0.9\columnwidth]{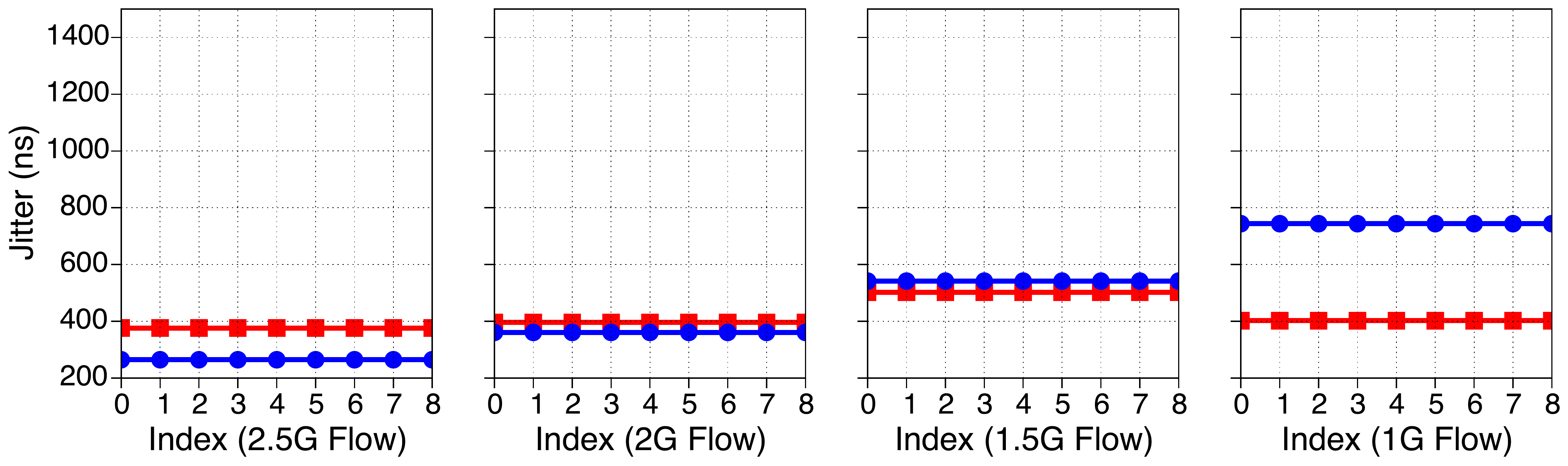}}
    \vspace{-5pt}
    \caption{Delay metrics for flows in a 36-radio setup. \algnameD
    meets the e2e guarantees of all flows.\label{fig:radio_delays}}
\end{figure}

We use the NS--3 packet tagging mechanism to tag each baseband packet with a
priority level. The priority levels are chosen to achieve e2e schedulability
(Sec.~\ref{sec:guarantees}). We then implement a packet-classifier at the edge switch
that classifies the incoming packet and places it in one of the 4 egress queues. For
dequeuing, the queues are searched in decreasing order of their priority.
For simplicity, we assume the transport delay-bound is same
as the inter-arrival period of each flow. Under this assumption, the priorities
are determined by the inverse of the period (equivalent to a rate-monotonic scheduler).

Fig.~\ref{fig:radio_delays} shows the maximum e2e transport delays and
jitter observed in a 3-$ary$ tree with 36 radios. Here, we compare \algnameD's
scheduling with the basic FIFO packet scheduling. As seen from Fig.~\ref{fig:radio_delays_3},
\algnameD meets the e2e guarantees (delay $<$ inter-arrival time) of each flow while
the FIFO scheduling misses the delay bound of the 2.5G flow. In prioritized scheduling,
the 1G flow has the least priority and therefore sees an increase in its transport delay.

\begin{figure}[t]
    \centering
    \includegraphics[width=0.8\columnwidth]{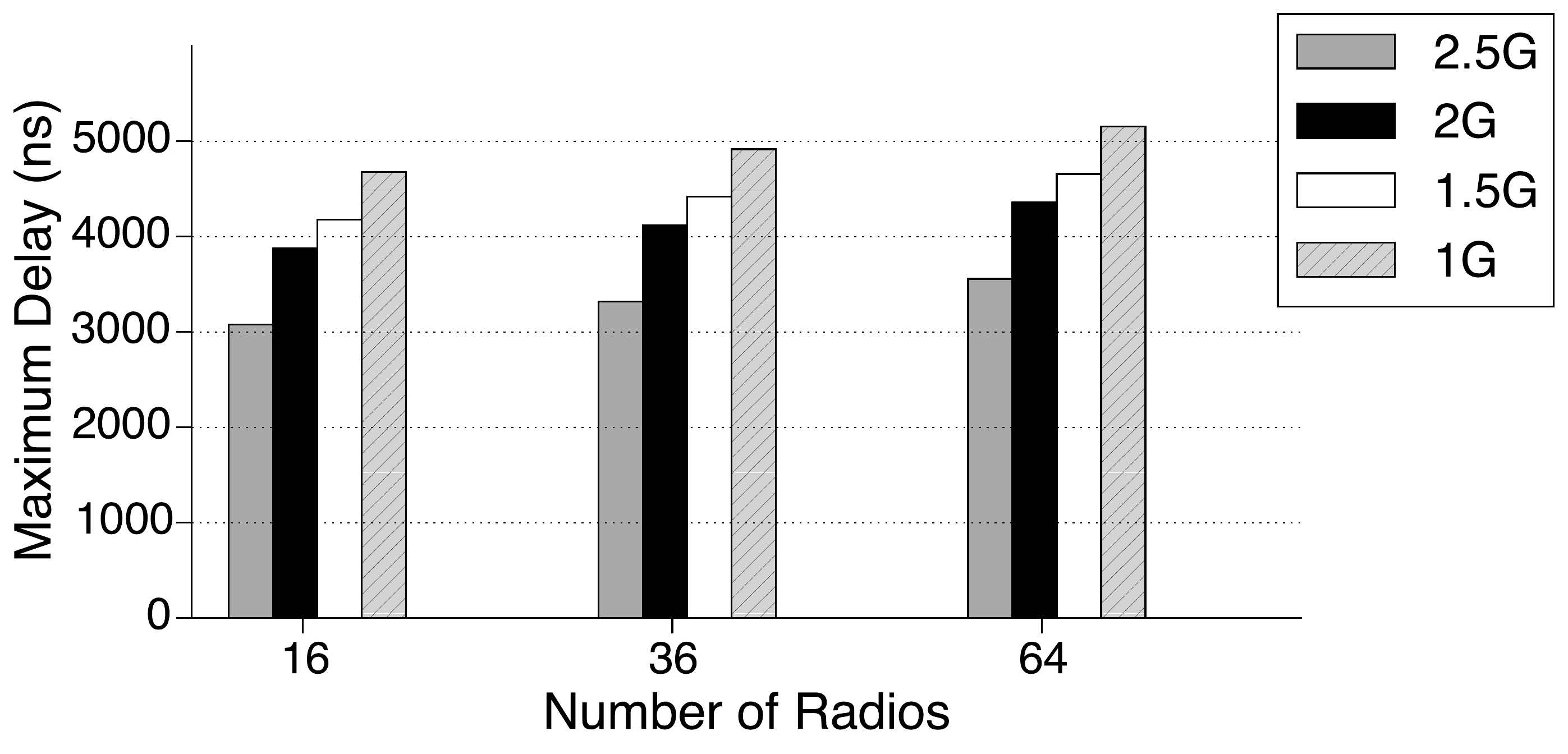}
    \caption{Maximum e2e delay of baseband traffic flows with increasing scale of the network.\label{fig:scale_delay}}
\end{figure}

In Fig.~\ref{fig:scale_delay}, we show the maximum observed flow delays as we increase
the number of radios in the network. Assuming a fixed number of radios per
edge switch, the scaling is achieved by using $q^2$ edge switches, and $q$ aggregation
switches for a $q$-ary aggregation tree. The core switch capacity, however, is
fixed at 200Gbps. Despite quadrupling the number of radios (from 16 to 64),
the maximum delay in the network increases not more than 300$ns$ (a less than 10\%
increase). This is attributed to the maximum delay bound that is dependent only
on the height of the tree, and not on the number of radios. Note that this
constant increase is not always the case, for example, in daisy-chained radio
network in ARGOS \cite{argos}, the maximum delay grows linearly with the number of radios.

\section{Achievable Capacity Under E2E Schedulability}\label{application}

In this section, we characterize the achievable wireless capacity under the
constraint of e2e schedulability of baseband transport. We consider a basic
implementation of \algnameD with single aggregation switch, $q$ edge switches,
and MIMO decoding of spatially-multiplexed baseband samples from $n$ radios.
The specifics of synchronization, buffering and decoding are omitted but are
assumed to manifest through the requested e2e delay bounds (Sec.~\ref{delay_bound}).


Assuming everything else is fixed, the ADC quantization determines the flow rate
of a radio. It also affects the wireless capacity through quantization noise and
thus is the search parameter to maximize capacity while ensuring schedulability.

Fig.~\ref{fig:radio_model} shows our system model where samples received by radios
over a wireless channel are transported for decoding. Let $\ve{x} \in \mathbb{C}^{m\times 1}$
be the signal vector sent by the transmitter, and let $\ve{y} \in \mathbb{C}^{n\times 1}$
be the received signal vector, where $m$ is the number of transmit antennas,
and $n$ is the number of radios. Let $\ve{H} \in \mathbb{C}^{n\times m}$ represent
the wireless channel between the transmitter and the radios.

\begin{figure}[t]
    \centering
    \includegraphics[width=0.7\columnwidth]{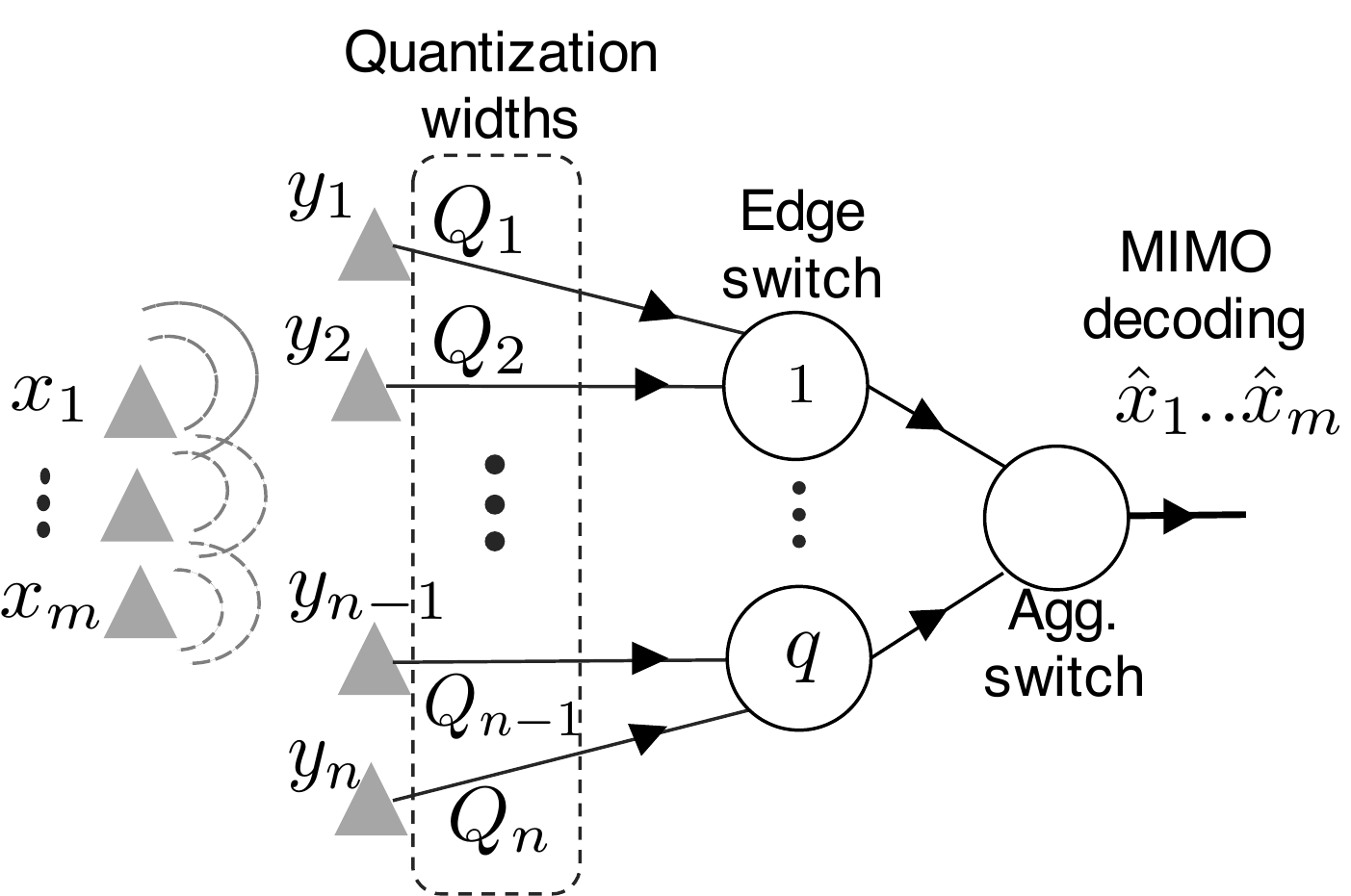}
    \caption{Model for calculating wireless capacity with e2e scheduling using
    ADC quantization width as search parameter.\label{fig:radio_model}}
\end{figure}

\textbf{Quantization}. Assume radio $i$ has ADC quantization width,
$Q_i$, that takes an integer value from the set $\mathcal{L} = \{L_1, \ldots, L_d\}$,
where $L_1$ and $L_d$ are the minimum and the maximum quantization widths, respectively.
Further, let $\ve{Q} = [Q_1, \ldots, Q_n]$, and let $\gamma(Q_i)$ be the average quantization
noise power injected into the baseband samples, where $\gamma(\cdot)$ is the quantization
noise function.

\textbf{Transport}. Assume a fixed size $B$ of baseband packets. From Eq.~\eqref{eq:distro:rate} and
Eq.~\eqref{eq:distro:arrival}, the inter-arrival time at radio $i$ is $T_i = \frac{B}{2Q_if_i}$, which
is a function of the radio quantization width, $Q_i$.

\textbf{Model}. Assuming a narrow-band channel, the received signal can be expressed as:
\begin{equation}\label{eq:sig_mimo}
\ve{y} = \ve{H}\ve{x} + \ve{z} + \ve{z}_Q,
\end{equation}
where $E[\ve{x} \ve{x}^H] = \rho \ve{I}_{m}$, $\rho$ denotes the transmitted power that
is equal across all symbols, $\ve{z} \sim \mathcal{C}\mathcal{N}(0,\sigma^2\ve{I}_{n})$ is
the additive complex white gaussian noise, and $\ve{z}_Q$ represents the quantization noise vector.
We assume the effect of quantization manifests through the additive term $\ve{z}_{Q}$, which is
approximated as a zero-mean complex Gaussian noise with covariance
matrix $\ve{\Sigma}_Q = \text{diag}(\gamma(Q_i), \ldots, \gamma(Q_n))$.
Let $\ve{\Sigma} = \sigma^2\ve{I}_{n} + \ve{\Sigma}_Q$ denote the equivalent noise covariance matrix.

The ergodic wireless capacity (in b/s/Hz) for a fixed $\ve{Q}$, with imperfect
channel knowledge at the transmitter, is \cite[Eq. 20]{mimo_capacity}:
\begin{align}\label{eq:cap_mimo}
R(\ve{Q}) = \mathrm{E}_{\ve{H}}[\log_2 \text{det} (\ve{I} + \rho\ve{\Sigma}^{-1}\ve{H}\ve{H}^{H})]
\end{align}
where the expectation $\mathrm{E}_{\ve{H}}[\cdot]$ is taken over all channel
realizations of $\ve{H}$.

\subsection{Problem Formulation}

We want to select the quantization, $\ve{Q}$, that ensures e2e schedulability,
and also maximizes the wireless capacity. This is expressed through an optimization
problem (OP):
\begin{align}
\text{OP:} \quad &\quad\max_{\ve{Q} \in \mathcal{L}^{n}} R(\ve{Q})\\
\text{s.t.}\quad &\tau_i(\ve{Q}) = (T_i, D_i), \forall i, \text{is schedulable}\label{eq:schedulability}
\end{align}
where $D_i$ is the transport delay bound of radio $i$. Solving OP is not practical
as there exists no e2e schedulability test for Eq.~\eqref{eq:schedulability}.
Instead, we solve the problem OP' for the restricted set of flows used in \algnameD 
for which schedulability test exists at the edge switches, that is:  
\begin{align}
\text{OP':} \quad &\quad\max_{\ve{Q} \in \mathcal{L}^{n}} R(\ve{Q})\\
\text{s.t.}\quad &(T_i, d_i^{'}), \forall i \in S_k, \forall k,
\text{is schedulable} \label{eq:schedulability_restrict}
\end{align}
where $d_i^{'}$ is defined as in Theorem~\ref{theorem:distro}.
Next, we show the following properties for the optimization problem.

Assuming the quantization noise is only dependent on the number of quantization 
bits, its power generally decreases with increasing resolution (e.g. $\gamma(x) \propto 2^{-2x}$
for uniform-level quantization).  

\begin{myprop}\label{prop:mono_rate}
If $\gamma(\cdot)$ is a monotonically decreasing function, then $R(\ve{Q})$ is 
monotonically increasing in $\ve{Q}$.
\end{myprop}

\begin{proof}
Let $\ve{Z} = \ve{I} + \rho\ve{\Sigma}^{-1}\ve{H}\ve{H}^{H}$, $\ve{\Sigma}^{-1} = \text{diag}(\mu_1, \ldots, \mu_n)$,
where $\mu_i = 1/(\sigma^2+\gamma(Q_i))$ for $i = 1, \ldots, n$, and
let $\lambda_1, \ldots, \lambda_n$ denote the eigen values of $\ve{Z}$
such that $\lambda_1 \geq \lambda_2 \geq \ldots \geq \lambda_n$.
We follow an approach similar to \cite[Theorem 10.3]{matrix} for the proof.
For a fixed $i$, we claim that $\lambda_j$, for every $j$, is monotonically increasing
in $\mu_i$. On the contrary, suppose this is not true and there exists an
interval of $\mu_i$, and $k$, $k \in \{1,\ldots, n\}$, such that $\lambda_k$ increases
and decreases in that interval. Consequently there exists $\lambda^{'}$ such that
$(\lambda^{'}-\lambda_k)$ vanishes for at least two values of $\mu_i$.
Since $\text{det}(\lambda\ve{I} - \ve{Z})$ $=\Pi_{j=1}^{n}(\lambda - \lambda_j )$,
then, for $\lambda = \lambda^{'}$, $\text{det}(\lambda\ve{I} - \ve{Z})$
vanishes for more than one value of $\mu_i$, which is impossible,
since $\text{det}(\lambda\ve{I} - \ve{Z})$ is a linear polynomial in $\mu_i$.
Therefore if $\gamma(Q_i)$ strictly decreases with $Q_i$,
$\mu_i$ strictly increases with $Q_i$, thus, $\forall j$, $\lambda_j$
monotonically increases with $Q_i$, for any $i = 1, \ldots, n$.
Since $\log_2\text{det}(\ve{Z}) = \sum_{j=1}^{n} \log_2\lambda_j$,
and expectation is a monotonic operator, therefore, $R(\ve{Q}) = \mathrm{E}_{\ve{H}}[\log_2\text{det}(\ve{Z})]$ 
is monotonically increasing with $Q_i$, $\forall i$.
\end{proof}

\begin{myprop}\label{prop:mono_sched}
For quantization $\ve{Q}$, if the set of traffic flows, $(T_i, d_i^{'}), \forall i \in S_k$
is schedulable for all edge switches $k$, where $d_i^{'}$ is as defined in Theorem~\ref{theorem:distro},
then the traffic flows $(T_i^{'}, d_i^{'}), \forall i \in S_k$ are also schedulable for 
all $k$, for quantization, $\ve{Q}^{'}$, such that $Q_i^{'} \leq Q_i, \forall i$.
\end{myprop}

\begin{proof}
Let $\mathcal{K}^{'}$ be the set of edge switches where radios see the reduction
in their quantization width. Consider any $ k \in \mathcal{K}^{'}$: for all $i \in S_k$,
changing $Q_i$ to $Q_{i}^{'}$ increases the inter-arrival period from $T_i$ to $T_i^{'}$.
For simplicity, assume the switch implements EDF scheduling and therefore
Eq.~\eqref{eq:edf2} is true for flows $(T_i, d_i^{'}), i \in S_k$. Since $T_i < T_i^{'}$ implies
$\lceil{(t-d_i^{'})/T_i^{'}}\rceil^{+} \leq \lceil{(t-d_i^{'})/T_i}\rceil^{+}$,
$\forall t$, substituting in Eq.~\eqref{eq:edf2}, it follows that $(T_i^{'}, d_i^{'})$,
$i = 1,\ldots,n$, is also schedulable using Theorem~\ref{theorem:edf}. A similar 
exercise with Eq.~\eqref{eq:fixed2} shows that the flows are also schedulable for
fixed-priority non-preemptive scheduling given by Theorem~\ref{theorem:fixed}. 
\end{proof}

\subsection{Search Algorithm}

The search space of OP$'$ is $\mathcal{L}^n$, hence, finding the optimal $\ve{Q}$
through a brute-force search has an exponential complexity.
Note that we cannot relax the integral constraint because of the discreteness of
the schedulability. However, from the monotonic dependence on $\ve{Q}$ (Propositions
\ref{prop:mono_rate}--\ref{prop:mono_sched}), we can construct a greedy search algorithm.


The idea is to use breadth-first search (BFS) on the enumeration of the search
space. Starting from the highest quantization, $[L_d, \ldots, L_d]$, we enumerate
the next highest quantizations, $[L_{d-1}, L_{d}, \ldots, L_d]$,
$\ldots$, $[L_{d}, L_d, \ldots, L_{d-1}]$, and then enumerate the next highest quantization
for each of them, and so on. More generally, let us define the function $enum(\cdot)$,
for quantization,
$\ve{Q} = [L_{k_1}, L_{k_2}, \ldots, L_{k_n}]$, as follows:
\begin{align*}\label{eq:enum}
enum([L_{k_1}, L_{k_2}, \ldots, L_{k_n}]) = &\{[L_{k_1-1}, L_{k_2}, \ldots, L_{k_n}], \\
                                &[L_{k_1}, L_{k_2-1}, \ldots, L_{k_n}], \\
                                & \ldots, [L_{k_1}, L_{k_2}, \ldots, L_{k_n-1}]\}.
\end{align*}
As we enumerate each possible quantization, we check the flow shcedulability
(according to Eq.~\eqref{eq:schedulability_restrict}). If the flows are schedulable,
we calculate the corresponding capacity, but do not enumerate the quantization 
further. Finally, from the calculated capacities, we find quantization $\ve{Q}^\star$ that
achieves the maximum capacity.

Algorithm \ref{alg:search} gives the pseudo-code for implementing our proposed
solution. It uses a queue data structure, $U$, for breadth-first traversal. We
show that this algorithm finds a solution to OP$'$. Though in the worst
case it has the same $O(d^n)$ complexity, the running time in practice is
much less than a brute-force search.

\begin{algorithm}[t]
\caption{BFS search}\label{alg:search}
\begin{algorithmic}[1]
          \State{{\it Initialize}: $\ve{Q}^\star \leftarrow \phi$ , $C^\star \leftarrow 0$, $U \leftarrow  \text{Queue}()$}
          \State{{\it Returns}: $C^\star$, $\ve{Q}^\star$, optimal capacity and quantization}
          \State{U.push($[L_d, \ldots, L_d]$)}
          \While {U is not empty} \Comment{breadth-first traversal}
          \State{$\ve{Q}^{'} \leftarrow$   U.pop()}
          \If{$\ve{Q}^{'}$ is schedulable} \Comment{check e2e schedulability}
          \If{$R(\ve{Q}^{'}) > C^\star$}
          \State{$C^\star \leftarrow R(\ve{Q}^{'})$, $\ve{Q}^\star \leftarrow \ve{Q}^{'}$ }
          \EndIf
          \Else
          \For {$\ve{T}$ in $enum(\ve{Q}')$}
          \State{ U.push($\ve{T}$)} \Comment{next highest quantization}
          \EndFor
          \EndIf
          \EndWhile
  \end{algorithmic}
\end{algorithm}

\begin{theorem}\label{theorem:search}
Algorithm \ref{alg:search} gives the optimal solution to OP$'$.
\end{theorem}
\begin{proof}\let\qed\relax
We show that $\ve{Q}^\star$ is both schedulable and achieves maximum capacity.
Clearly, the schedulability of $\ve{Q}^\star$ holds from Line 6 in the algorithm. 
We know from Proposition \ref{prop:mono_rate}, the capacity from enumeration (in Line 11) 
always decreases, and hence $C^\star$ will always be larger than the capacity of 
all unenumerated quantizations. Note that from Proposition \ref{prop:mono_sched}), 
when a given quantization is schedulable, all its enumerations with smaller 
quantizations are also schedulable. Also, from Line 7, $C^\star$ is the maximum 
of all enumerated quantizations that are schedulable. Therefore, $C^\star$ is 
the optimal capacity that is schedulable.\qedwhite
\end{proof}

In summary, \algnameD's delay characteristics not only provides 
real-time guarantees but also supports the design of an optimally performing 
network.

\section{Related Work}\label{related}

Fat-Tree topology was first introduced as a routing network for parallel 
computation for networks-on-chip\cite{fat_tree}; its formal definition as "k--ary n--trees" 
is given in \cite{fat_tree_formal}. In the context of datacenter networks,
where commodity (off-the-shelf) switches are used, and have a fixed number 
of ports, fat-tree topologies (or folded clos networks \cite{clos}) have been used 
to build scalable and modular architectures \cite{alfaras, fat_tree_follow}. 
As these networks are bandwidth bound, they cannot meet real-time requirements. We note that 
the Fat-Tree variant used in this paper is the simplest kind
with no interconnections between the edge and core switches. 

A large-scale MU-MIMO was realized using a distributed architecture
of servers in \cite{bigstation:13}. The authors claim that bandwidth needed for
baseband transport is not an issue as modern switches support up to 40Gbps links.
However, they do not consider the real-time guarantees of transporting
baseband samples. ARGOS \cite{argos} is another practical multi-antenna setup
that suggests daisy-chained radios with Tree-based aggregation. However, no
real-time analysis was provided there. C-RAN \cite{cran} also proposes a 
data-center model for baseband processing, which must be carefully designed as 
these networks are known to have varying traffic and congestion patterns \cite{dctcp}.

To eliminate jitter, scheduled Ethernet using global scheduling has been proposed
for fronthaul transport in a C-RAN\cite{ashwood}.  On the other hand, the authors of \cite{spiro, uta}
studied the compression of baseband signals. They propose quantization schemes
for lossy compression of the baseband samples.
In \cite{spiro}, the authors propose prioritization of baseband frames, though,
the priorities are not based on packet delays and therefore cannot provide
e2e guarantees.

Real-time transport of Ethernet packets is a well-known problem in real-time systems
literature \cite{kandlur:91, hui:95, survey:94}. While \cite{kandlur:91} provides
a general approach for schedulability, the authors propose the use of flow regulators at
each switch, which is difficult to realize as compared to \algnameD's edge-switch 
scheduling.

\section{Conclusion}\label{conclusion}

This paper provides a design for baseband transport network based on
Fat-Tree topology that is scalable to a large number of radios while guaranteeing
real-time delivery. We calculate the delay bound of baseband traffic
and provide sufficient criteria for e2e schedulability, which is validated via 
simulations. Certain design restrictions such as symmetric links and equal packet 
sizes are made but may be relaxed as long as a bound for queuing time at each 
switch is obtainable. 


We also characterize the wireless capacity with the schedulability constraint,
and provide an efficient search algorithm to maximize it. In practice, our design 
and analysis provides a principled approach to the deployment of baseband transport
networks that are critical for upcoming wireless system architectures.

\section{Appendix}\label{appendix}
%
\begin{customthm}{A.1}\label{theorem:edf}
Under deadline scheduling, a set of traffic flows $\tau_i = (T_i, d_i), i = 1,2,\ldots, n$,
is schedulable on a link with transmission time $C$, and preemption if and only if \cite[Theorem 1]{zheng:94}:
\begin{equation}\label{eq:edf}
\forall t>0, \quad \frac{C}{t}\left(\sum_{i=1}^{n}\left\lceil{(t-d_i)/T_i}\right\rceil^{+}\right) \leq 1,
\end{equation}
and without preemption if and only if \cite[Theorem 6]{zheng:94}:
\begin{equation}\label{eq:edf2}
\forall t\geq d_{min}, \quad \frac{C}{t}\left(1 + \sum_{i=1}^{n}\left\lceil{(t-d_i)/T_i}\right\rceil^{+}\right) \leq 1
\end{equation}
where $d_{min} = \min\{d_i: 1\leq i \leq n\}$ and the function $\lceil x \rceil^{+} = max(0,\lceil x \rceil)$.
\end{customthm}

\begin{customthm}{A.2}\label{theorem:fixed}
Under fixed-priority scheduling, a set of traffic flows $\tau_i = (T_i, d_i), i = 1,2,\ldots, n$,
and priorities $\pi_i$, $i = 1,2,\ldots, n$,  such that $\pi_1 \geq \pi_2 \geq \ldots \geq \pi_n$, is schedulable
on a link with transmission time $C$, if the function $W_m (k,x)$ satisfies:
\begin{equation}
\max_{1\leq m\leq n} \max_{k \leq N_m} W_m (k, (k-1)T_m + d_m) \leq 1
\end{equation}
where $N_m = \min\{k : W_m(k, kT_m) \leq  1\}$, and function $W_m (k,x)$, for preemption, is defined as \cite{cmu:90}:
\begin{equation}
W_m(k,x) = \min_{0<t\leq x} \frac{C}{t}\left ( k+ \sum_{i=1}^{m-1} \left\lceil t/T_i\right\rceil \right),
\end{equation}
and without preemption, $W_m (k,x)$ is defined as \cite{utah:99}:
\begin{equation}\label{eq:fixed2}
W_m(k,x) = \min_{0<t\leq x} \frac{C}{t}\left ( k+1 + \sum_{i=1}^{m-1}\left(1 + \left\lfloor{(t-C)/T_i}\right\rfloor \right) \right)
\end{equation}
%
\end{customthm}

\balance
\small
\bibliographystyle{IEEEtran}
\bibliography{main}
\end{document}